\documentclass[onecolumn,amsmath,amssymb,floatfix]{revtex4-2}

\usepackage{graphicx}
\usepackage{dcolumn}
\usepackage{bm}
\usepackage{hyperref}
\usepackage{xcolor}
\usepackage{tikz}
\usepackage{pgfplots}
\pgfplotsset{compat=1.18}
\usepackage{booktabs}
\usepackage{multirow}
\usepackage{amsmath,amssymb,amsthm}
\usepackage{algorithm}
\usepackage{algpseudocode}
\usepackage{bbm}

\newcommand{\R}{\mathbb{R}}
\newcommand{\Z}{\mathbb{Z}}

\newcommand{\norm}[1]{\left\|#1\right\|}

\newtheorem{proposition}{Proposition}

\newcommand{\stimcirc}{%
  \texttt{surface\_\allowbreak code\allowbreak:\allowbreak
          rotated\_\allowbreak memory\_\allowbreak z}}

\begin{document}

\title{Physics-Informed Graph-Neural Decoding of the Surface Code:\\
the Logical Signal as an Exact Topological Pairing}
\author{P. E. Trevisanutto}
\thanks{Corresponding author}
\email{paolo.trevisanutto@stfc.ac.uk}
\affiliation{Scientific Computing Department, Science \& Technology Facilities Council, Daresbury Laboratory, Warrington, WA4 4AD, United Kingdom.}
\author{S. Dhanpal}
\affiliation{Scientific Computing Department, Science \& Technology Facilities Council, Rutherford Appleton Laboratory, Didcot OX11 0QX, United Kingdom.}
\author{S. Basak}
\affiliation{Scientific Computing Department, Science \& Technology Facilities Council, Rutherford Appleton Laboratory, Didcot OX11 0QX, United Kingdom.}
\author{L. Petit}
\affiliation{Scientific Computing Department, Science \& Technology Facilities Council, Daresbury Laboratory, Warrington, WA4 4AD, United Kingdom.}
\author{J.Thiyagalingam}
\affiliation{Scientific Computing Department, Science \& Technology Facilities Council, Rutherford Appleton Laboratory, Didcot OX11 0QX, United Kingdom.}
\date{\today}

\begin{abstract}
We develop a physics-informed graph neural network (GNN) decoder for the surface
code that solves a discrete Poisson equation on the syndrome graph, with the
syndrome as the charge source. We compare four readout architectures for
extracting the logical-error probability: a potential-based readout that maps the
Poisson field through a multilayer perceptron, two current-based readouts under
single- and two-sink Dirichlet boundary conditions, and a diffusion-based
variant. Comparing these, we show that the solver's edge current is a pure
gradient flow whose harmonic (circulating) part vanishes identically. The
logical signal therefore cannot be read as a component of the current itself; it
is instead a topological pairing between the syndrome and a boundary-fixed
harmonic coordinate that distinguishes the two code boundaries linked by the
logical operator. We prove
that this pairing is evaluated exactly and in closed form, with no learned
readout parameters, as the net current drained between the two boundary sinks. On
the rotated surface code under circuit-level depolarising noise, this single
closed-form scalar matches the best full-field readout and, at larger code
distance, significantly exceeds the single-sink current pool, so that isolating
the pairing helps more, not less, as the field grows larger and sparser. The
decoder is not intended to surpass minimum-weight perfect matching, near-optimal
for this noise model; its contribution is an interpretable characterisation of
the logical signal itself.
\end{abstract}


\maketitle

\section{Introduction}
\label{sec:intro}

Fault-tolerant quantum computing requires active quantum error correction
(QEC) \cite{Terhal2015}.  Among the many proposed QEC codes, the surface code~\cite{Kitaev2003}
occupies a privileged position. It has a high threshold error rate
($p_{\rm th}\approx 10.3\%$ under independent depolarizing
noise~\cite{Dennis2002}), requires only nearest-neighbour interactions, and
its syndrome extraction is planar. For these types of codes, the canonical decoder is
Minimum-Weight Perfect Matching (MWPM)~\cite{Edmonds1965,Dennis2002}: syndrome excitations are matched on a weighted graph 
whose weights encode the log-likelihood of each error path.  Under independent depolarizing noise with
known error rate $p$, MWPM is close to optimal because the decoding problem
maps exactly onto the random-bond Ising model on the Nishimori
line~\cite{Dennis2002,Nishimori1981}, where MWPM implements maximum-likelihood
inference.

Despite this near-optimality, MWPM has two important limitations.  First, it
requires explicit knowledge of the noise model (the physical error rate $p$ or
the full error channel) to set edge weights.  In practice, the noise model
drifts over time and may be difficult to characterise precisely.  Second,
at circuit level, where measurement errors, gate errors, and correlated
noise arise, MWPM with standard weights becomes sub-optimal, and the
construction of accurate weight tables is non-trivial.

Learned decoders address both limitations by inferring the decoding function
directly from data, without assuming a known noise
model~\cite{Torlai2017,Varsamopoulos2017,Krastanov2017,Baireuther2018,
Maskara2019,Meinerz2022,Bausch2023}.  Most existing neural decoders treat the
syndrome as an image or a flat vector, losing the geometric structure of the
syndrome graph.  Graph neural networks (GNNs) restore this structure and
have been applied to QEC with promising
results~\cite{Overwater2022,Lange_2023,Gicev2023}.

In this work, inspired by the recent Joo's article \cite{Joo2026}, we propose a physics-informed GNN whose architecture embeds a discrete Poisson equation as a hard inductive bias, and we identify which part of that current carries the logical-error signal through the Hodge decomposition of the resulting current. The use of graph Laplacians in machine learning has a long
history~\cite{Belkin2001,Zhu2003,Zhou2004}; the novel element here is
their use as a \emph{differentiable physics layer} inside a GNN decoder
for a quantum error-correcting code, with the incidence matrix $\mathbf{B}$
playing the role of the discrete exterior derivative.
Our key contributions are:
\begin{enumerate}
  \item \textbf{An exact, closed-form readout (Prop.~\ref{prop:pairing}).}
    We prove that the logical-error signal is the pairing of the syndrome with
    the harmonic representative of the generator of the relative cohomology
    $H^1(G,\partial G;\Z_2)$, and that it is evaluated exactly---with no learned
    readout parameters---as the net current between two boundary sinks. The
    logical content is a topological invariant, read in closed form rather than
    approximated by a trained pooling stage. A single boundary node cannot do
    this: it merges the two code boundaries, no separating harmonic function
    exists, and the pairing is undefined.
  \item \textbf{The logical signal is boundary data, not a current component.}
    We show the solver's current is a pure gradient flow whose harmonic part
    vanishes identically, so no part of the field carries the logical class.
    This places the readout on the discrete de~Rham complex: the incidence
    matrix $\mathbf{B}$ is the discrete exterior derivative $d_0$, and the
    weighted Laplacian
    $\mathbf{L}=\mathbf{B}^\top\mathrm{diag}(\bm{w})\mathbf{B}$ is the discrete
    Hodge Laplacian $\Delta_0$, so the discrete Gauss law
    $\mathbf{L}\bm{\phi}=\bm{\rho}$ is an \emph{exact} statement, not an
    approximation, and the GNN learns only the metric $\bm{w}$.
  \item \textbf{A controlled empirical validation.} In a like-for-like ablation
    (single-sink vs.\ two-sink, identical encoder, solver, and training,
    differing only in boundary structure and readout) with multi-seed
    statistics, the closed-form scalar matches the best full-field readout at
    $d=5$ and significantly exceeds the single-sink pool at $d=7$ (paired
    $t=4.1$, $p=0.014$). The parameter-free topological readout thus saturates
    the logical information the full-field readouts access, exactly as a
    one-dimensional logical signal ($k=1$) predicts.
\end{enumerate}

\section{Background}
\label{sec:background}

\subsection{Surface Code and Syndrome Graph}
\label{sec:surface_code}

The rotated surface code of distance $d$ encodes one logical qubit in $d^2$
data qubits~\cite{Fowler2012}.  Syndrome measurements project data-qubit
errors onto a set of $(d^2-1)$ binary syndrome bits: $(d^2-1)/2$ detectors
measure $Z$-type (bit-flip) errors, and $(d^2-1)/2$ measure $X$-type
(phase-flip) errors.  A Pauli error $E$ is detected as the set of stabilisers
that anti-commute with $E$; the syndrome is the binary vector
$\bm{\sigma}\in\{0,1\}^{d^2-1}$.

The \emph{syndrome graph} $G=(V,E)$ (Fig.~\ref{fig:surface_code_sink}) has one
node per syndrome bit, plus a virtual boundary node $v_b$ that absorbs error
chains reaching the code boundary.  An edge $(u,v)\in E$ represents a single-qubit error path whose
endpoints are the two syndrome bits it excites; its weight is the
log-likelihood $c_e = -\log[p_e/(1-p_e)]$.  Decoding reduces to finding the
minimum-weight perfect matching on the subgraph induced by the set of
\emph{excited} syndrome bits (those with $\sigma=1$).

This work uses the rotated surface code under circuit-level noise
(Fig.~\ref{fig:surface_code}).  All experiments use the
\stimcirc{} circuit generated by
\textsc{Stim}~\cite{Gidney2021}, which defines a single logical observable
($\bar Z$); the syndrome graph is built from the corresponding circuit-level
detector error model (DEM).  Because each detector is a space-time event
(a stabiliser measurement at one syndrome round), the temporal dimension is
compiled directly into the graph geometry: a distance-$d$ code run for $r$
rounds yields $N=(d^2-1)\,r/2 + 1$ nodes once the temporal detector copies and
the virtual boundary are included.  Our main experiments use the distance-$5$
code with $r=5$ rounds, giving $N=121-2$ nodes (120 detectors plus one or two boundaries) and distance-$7$
code with $r=7$ rounds, giving $N=337-8$ nodes (336 detectors plus one or two boundaries);
we focus the quantitative comparison on $p=0.005$, the regime in which the distance-$5$ logical signal is well resolved, and, unless stated
otherwise, report statistics over five independent seeds, each with its own
train/validation/test split.  Node coordinates $(t,x,y)$ are retained so that
the boundary partition used by the two-sink readout (Sec.~\ref{sec:twosink})
can be defined geometrically.

\begin{figure}[t]
\centering
\includegraphics[width=\textwidth]{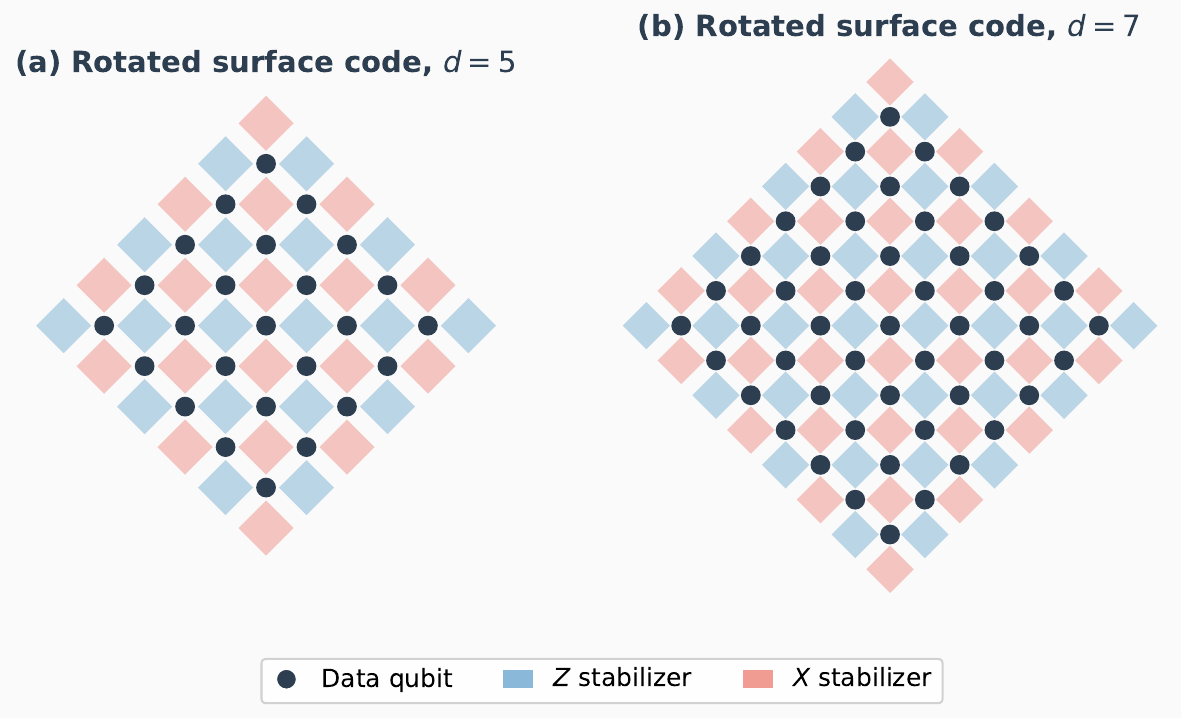}
\caption{%
  Rotated surface code and syndrome graph.
  (a)~Rotated surface code, $d=5$: data qubits (dark circles) on a
  $5\times 5$ grid; $Z$-type stabilisers (blue diamonds) and $X$-type
  stabilisers (red diamonds) alternate in a checkerboard pattern; this is
  the code used throughout.
  (b)~A larger rotated code ($d=7$) shown only to illustrate how the layout
  scales with distance.
}
\label{fig:surface_code}
\end{figure}

\begin{figure}[t]
\centering
\includegraphics[width=\textwidth]{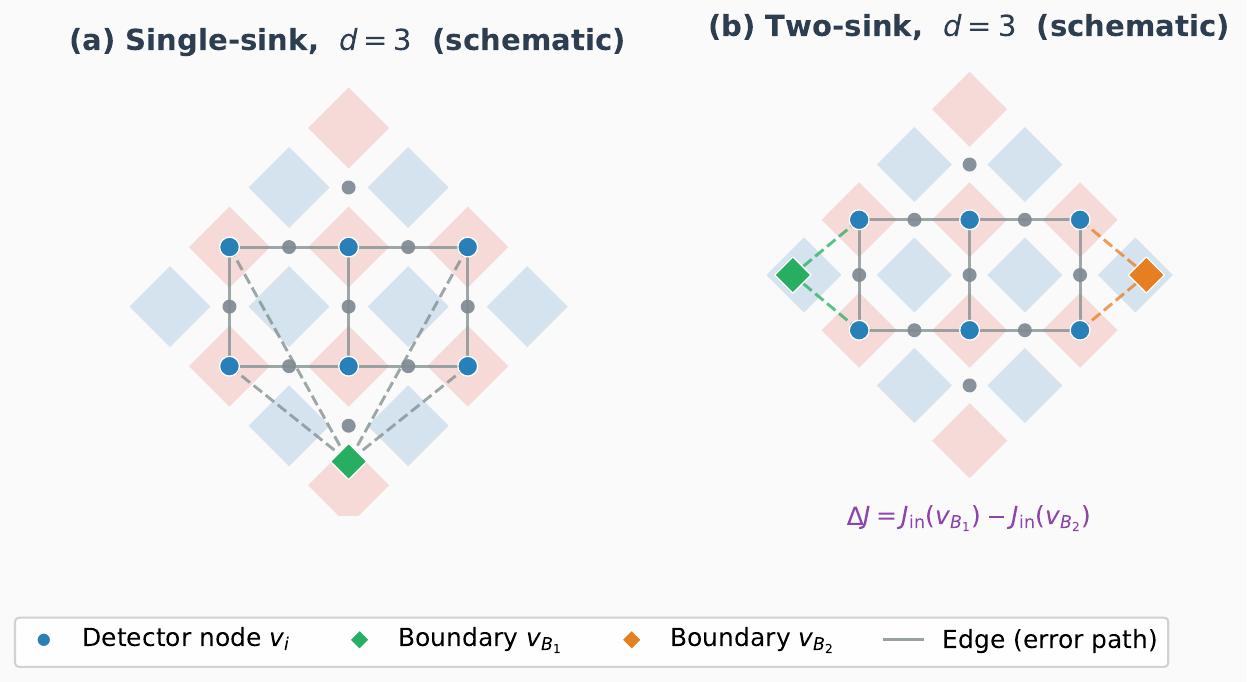}
\caption{%
  Syndrome graph $G$ (schematic): detector nodes (blue circles) sit at
  stabiliser plaquettes, and an edge connects two nodes if a single data-qubit
  error excites both.  Boundary sinks (diamonds) absorb error chains that reach
  the code boundary via dashed edges.
  (a)~The single-sink architecture (our baseline) uses one boundary node
  $v_b$, into which all chains drain.
  (b)~The two-sink architecture (used throughout) places one sink on each of the
  two code boundaries linked by the logical operator, $v_{b}^{(1)}$ and
  $v_{b}^{(2)}$; the readout is the net current between them,
  $\Delta J=J_{\rm in}(v_{b}^{(1)})-J_{\rm in}(v_{b}^{(2)})$
  (Sec.~\ref{sec:twosink}), read at the boundary rather than across any cut
  through the bulk (Fig.~\ref{fig:physics_current}).
}
\label{fig:surface_code_sink}
\end{figure}

\subsection{Minimum-Weight Perfect Matching}
\label{sec:mwpm}

MWPM~\cite{Edmonds1965} finds the perfect matching $M^*$ that minimises the
total edge weight among all matchings of the excited syndrome nodes.  Under
independent depolarizing noise with known $p$, this is equivalent to
maximum-likelihood decoding~\cite{Dennis2002}.  After matching, the parity of
the number of matched paths crossing the logical operator determines the
logical error prediction.

The computational complexity of exact MWPM is $O(n^3)$ in the number of
excited nodes $n$; in practice, the near-linear
\textsc{PyMatching}~\cite{Higgott2023} library is used.

\subsection{Discrete Exterior Calculus on Graphs}
\label{sec:dec}

Discrete Exterior Calculus (DEC)~\cite{Hirani2003,Desbrun2005} provides a
coordinate-free calculus on simplicial complexes.  On a graph
$G=(V,E)$ with an arbitrary edge orientation, the 0-forms are scalar
functions on nodes ($\R^N$) and the 1-forms are functions on oriented edges
($\R^{|E|}$).

The \emph{discrete exterior derivative} $d_0: \R^N \to \R^{|E|}$ is the
signed incidence matrix $\mathbf{B}\in\{-1,0,+1\}^{|E|\times N}$:
\begin{equation}
  (\mathbf{B})_{e,(u,v)} = \begin{cases}
    +1 & \text{if node }v\text{ is the head of edge }e=(u,v),\\
    -1 & \text{if node }u\text{ is the tail of edge }e,\\
    0  & \text{otherwise.}
  \end{cases}
\end{equation}
Given a weight 1-form $\bm{w}\in\R^{|E|}_{>0}$, the \emph{discrete
Hodge star} on 1-forms is $\star_1 = \mathrm{diag}(\bm{w})$.  The discrete
Hodge Laplacian on 0-forms is then
\begin{equation}
  \mathbf{L} = d_0^\top\, \star_1\, d_0
             = \mathbf{B}^\top \mathrm{diag}(\bm{w})\, \mathbf{B},
\label{eq:laplacian}
\end{equation}
which is the weighted graph Laplacian \cite{Joo2026}.  The discrete Gauss law
$\mathbf{L}\bm{\phi}=\bm{\rho}$ is the exact discrete analogue of
$-\nabla\cdot(w\,\nabla\phi)=\rho$ on a Riemannian manifold.  No
approximation is involved: the identity $d_0^\top \star_1 d_0$ holds exactly
on the graph.

The 1-form (flux) associated with the solution $\bm{\phi}$ is
\begin{equation}
  \bm{J} = \star_1\, d_0\,\bm{\phi}
          = \mathrm{diag}(\bm{w})\,\mathbf{B}\bm{\phi},
\label{eq:current}
\end{equation}
the weighted edge current, satisfying $\mathbf{B}^\top \bm{J} = \bm{\rho}$
(discrete Gauss law / Kirchhoff current law).
Fig.~\ref{fig:physics_current} illustrates this on a schematic
six-node syndrome graph with two boundary sinks: the excited detectors act as
sources, the current drains to the two sinks, and Kirchhoff's law is satisfied
exactly at every internal node.

\begin{figure}[t]
\centering
\includegraphics[width=\textwidth]{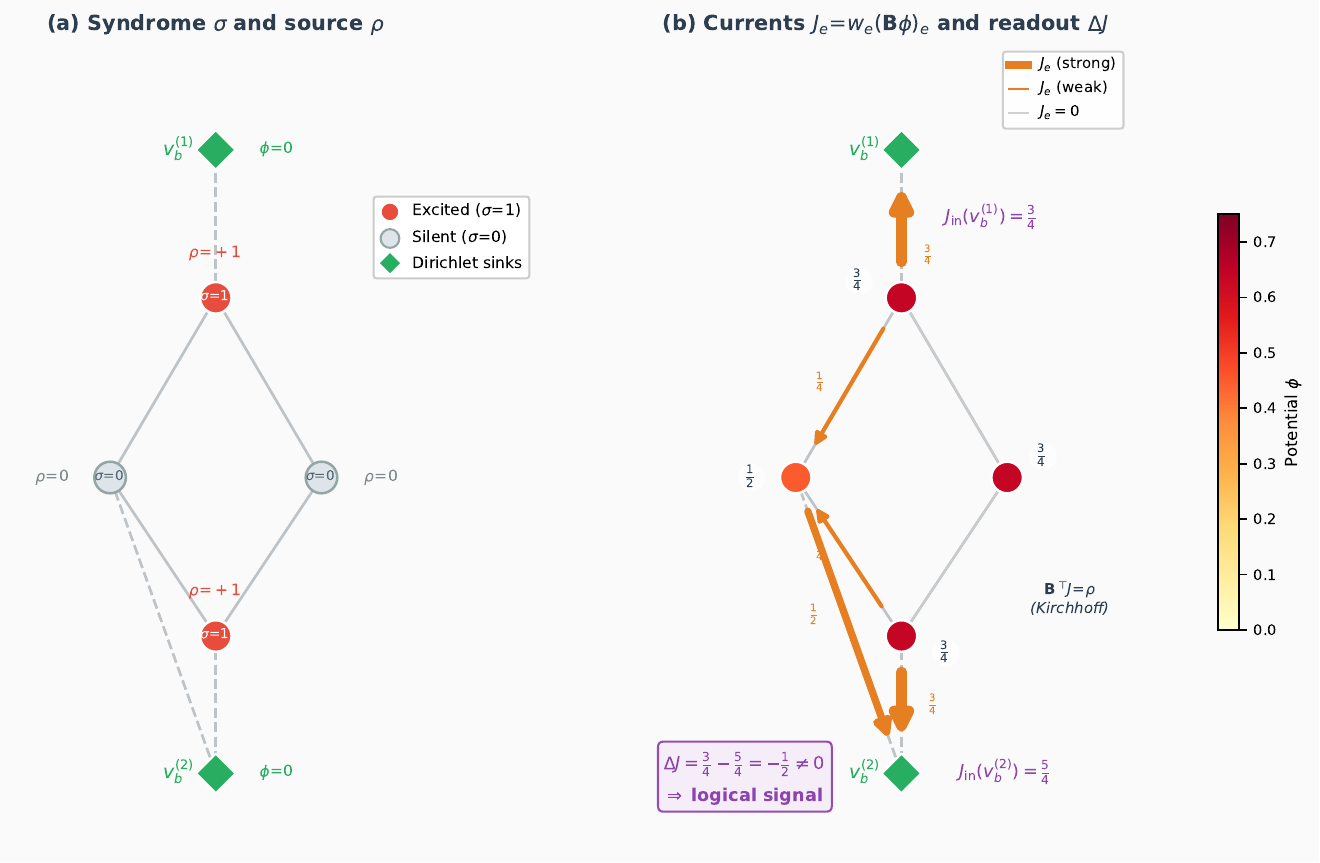}
\caption{%
  Physical interpretation of the two-sink readout on a schematic six-node
  syndrome graph (four detectors $v_0$--$v_3$ and two boundary sinks
  $v_b^{(1)}, v_b^{(2)}$), computed exactly with $w_e=1$.
  (a)~Two excited detectors ($\sigma=1$, red) act as sources $\rho=+1$; silent
  detectors have $\rho=0$.  Both sinks are Dirichlet nodes held at
  $\phi_{v_b^{(1)}}=\phi_{v_b^{(2)}}=0$, so no background charge is imposed on
  the detectors.
  (b)~Solving $\mathbf{L}(\bm{w})\bm{\phi}=\bm{\rho}$ yields the node potentials
  $\phi$ (colour scale); the edge currents $J_e=w_e(\mathbf{B}\bm{\phi})_e$ flow
  from high to low potential (arrow thickness $\propto|J_e|$), and Kirchhoff's
  law $\mathbf{B}^\top\bm{J}=\bm{\rho}$ holds at every internal node.  The
  readout is the difference of the currents drained by the two sinks,
  $\Delta J=J_{\rm in}(v_b^{(1)})-J_{\rm in}(v_b^{(2)})=\tfrac34-\tfrac54
  =-\tfrac12\neq0$, signalling a logical error.  This is \emph{not} a flux
  across any horizontal line: a cut would return the enclosed charge (a
  gradient quantity), whereas $\Delta J$ is the sink-current difference that
  Prop.~\ref{prop:pairing} identifies with the harmonic pairing
  $\sum_{\sigma_v=1}(1-2\chi_v)$ (Fig.~\ref{fig:harmonic_vote}).  A single
  boundary node would merge the two sinks, leaving no separating coordinate
  $\bm{\chi}$ and no pairing to read (Sec.~\ref{sec:hodge}).
}
\label{fig:physics_current}
\end{figure}

\subsection{Hodge Decomposition and the Logical Signal}
\label{sec:hodge}

The construction above returns a current $\bm{J}\in\R^{|E|}$, but not every
part of $\bm{J}$ carries logical information.  The Hodge (Helmholtz)
decomposition on the graph makes this precise.  On a $1$-dimensional complex
(nodes and edges only, no $2$-cells) every $1$-form splits orthogonally as
\begin{equation}
  \bm{J} \;=\; \underbrace{\mathbf{B}\,\bm{\psi}}_{\text{exact (gradient)}}
        \;\oplus\; \underbrace{\bm{J}_{\rm harm}}_{\text{harmonic}},
  \label{eq:hodge}
\end{equation}
where the exact part is a gradient of some potential $\bm{\psi}$ and the
harmonic part satisfies $\mathbf{B}^\top\bm{J}_{\rm harm}=\bm{0}$ while not
being a gradient.  The harmonic subspace is the kernel of the Hodge Laplacian
restricted to the non-exact sector; by the discrete Hodge theorem its dimension
equals the first Betti number $b_1=\dim H_1(G;\Z_2)$, the number of independent
cycles of the syndrome graph.  Most of these cycles are local and carry no
logical content; the logically relevant subspace is the homology relative to the
code boundary, $H_1(G,\partial G;\Z_2)$, whose dimension is the number $k$ of
logical qubits (Sec.~\ref{sec:onedim}).  For the planar rotated code
$k=1$: there is exactly one logical cycle, the one that links the two code
boundaries connected by $\bar Z$.

\textbf{The solver's current is exact, not harmonic.}

The current the solver produces is $\bm{J}=\mathrm{diag}(\bm{w})\mathbf{B}\bm{\phi}$
(Eq.~\eqref{eq:current}), i.e.\ a weighted gradient of the node
potential $\bm{\phi}$.  A gradient has no harmonic part: applying
$\mathbf{B}^\top$ gives $\mathbf{B}^\top\bm{J}=\mathbf{L}(\bm{w})\bm{\phi}=\bm{\rho}$,
so in the decomposition~\eqref{eq:hodge} (taken in the metric
$\star_1=\mathrm{diag}(\bm{w})$) $\bm{J}$ is purely the exact term and its
harmonic part vanishes,
\begin{equation}
  \bm{J}_{\rm harm} \;=\; \bm{J} - \mathrm{diag}(\bm{w})\mathbf{B}\bm{\phi} \;=\; \bm{0},
  \label{eq:jharm_zero}
\end{equation}
identically, for every syndrome.  This is not a defect but an elementary fact
about gradient flows: a resistor network driven by injected current, with no
flux threading it, carries no circulating current.  The solver's current is exact without any harmonic component $\bm{J}_{\rm harm}$ of
$\bm{J}$ to read.

The logical content must therefore be defined differently.  It is not a piece
of the current $\bm{J}$, but a number obtained by combining \emph{two} objects:
the syndrome (where the detectors fired) and a second field, fixed by the two
boundary sinks, that measures which side of the code each detector sits on.
Combining a measurement field with a source in this way is a \emph{pairing},
and the rest of this section constructs the field and proves that the two-sink
current evaluates the pairing exactly.

\textbf{The harmonic representative.}
Let $\partial G=\{v_b^{(1)},v_b^{(2)}\}$, one sink on each of the two code
boundaries linked by $\bar Z$, and let $I=V\setminus\partial G$.  Among all
$0$-cochains separating the two boundaries there is one that is
harmonic with respect to the metric $\bm{w}$:
\begin{equation}
  \bigl(\mathbf{L}(\bm{w})\bm{\chi}\bigr)_v = 0 \quad (v\in I),
  \qquad \chi_{v_b^{(1)}}=0, \qquad \chi_{v_b^{(2)}}=1 .
  \label{eq:chidef}
\end{equation}
Its differential $d_0\bm{\chi}$ is the harmonic representative of the generator
of the relative cohomology $H^1(G,\partial G;\Z_2)$ in the metric fixed by
$\star_1=\mathrm{diag}(\bm{w})$.  (Throughout, we use relative cohomology
$H^1(G,\partial G;\Z_2)$ for the harmonic representative $d_0\bm{\chi}$ against
which the syndrome is paired, and relative homology $H_1(G,\partial G;\Z_2)$ for
the error chains and the count $k$ of logical qubits; the readout is the pairing
between the two, and on a finite complex the two groups are isomorphic, so $k$ is
common to both.)  Topology fixes that there
is one such representative for the planar rotated code, since $k=1$; the metric
fixes \emph{which function} it is.  This is the precise sense in which the learned
weights matter: they do not select a component of a field, they deform the
representative against which the syndrome is measured.  The dimension of the
logical signal is a topological invariant; the representative that makes it
legible under a given noise model is exactly what MWPM obtains from the error
channel and what the GNN must learn from data.  The learned weights
$\bm{w}(\bm{\sigma})$ are the physical
content of the trained decoder.

\textbf{The readout evaluates the pairing exactly.}
The following identity is the theoretical basis for the two-sink readout of
Sec.~\ref{sec:twosink}.

\begin{proposition}[Sink current as a relative pairing]
\label{prop:pairing}
Let $\bm{\chi}$ solve Eq.~\eqref{eq:chidef} and let $\bm{\phi}$ solve the
Dirichlet problem
\begin{equation}
  \bigl(\mathbf{L}(\bm{w})\bm{\phi}\bigr)_v = \rho_v \quad (v\in I),
  \qquad \phi_{v_b^{(1)}} = \phi_{v_b^{(2)}} = 0 .
  \label{eq:phidef}
\end{equation}
With $\bm{J}=\mathrm{diag}(\bm{w})\mathbf{B}\bm{\phi}$, the net current drained
by each sink is
\begin{equation}
  J_{\rm in}\bigl(v_b^{(2)}\bigr) = \sum_{v\in I}\chi_v\,\rho_v,
  \qquad
  J_{\rm in}\bigl(v_b^{(1)}\bigr) = \sum_{v\in I}(1-\chi_v)\,\rho_v ,
  \label{eq:sinkcurrents}
\end{equation}
and therefore, with $\rho_v=\sigma_v$,
\begin{equation}
  s \;\equiv\; J_{\rm in}\bigl(v_b^{(1)}\bigr) - J_{\rm in}\bigl(v_b^{(2)}\bigr)
    \;=\; \sum_{v\,:\,\sigma_v=1}\bigl(1-2\chi_v\bigr).
  \label{eq:pairing}
\end{equation}
\end{proposition}
We refer to the weight $1-2\chi_v$ that each detector carries as its
\emph{vote}.
\begin{proof}
$\mathbf{L}=\mathbf{L}(\bm{w})$ is symmetric, so
$\bm{\chi}^\top\mathbf{L}\bm{\phi}=\bm{\phi}^\top\mathbf{L}\bm{\chi}$.
Splitting each side into the interior nodes $I$ and the two sinks,
\[
  \bm{\chi}^\top\mathbf{L}\bm{\phi}
  = \sum_{v\in I}\chi_v(\mathbf{L}\bm{\phi})_v
  + \chi_{v_b^{(1)}}(\mathbf{L}\bm{\phi})_{v_b^{(1)}}
  + \chi_{v_b^{(2)}}(\mathbf{L}\bm{\phi})_{v_b^{(2)}} .
\]
On $I$, $(\mathbf{L}\bm{\phi})_v=\rho_v$ by Eq.~\eqref{eq:phidef}; at the sinks
$\chi_{v_b^{(1)}}=0$ and $\chi_{v_b^{(2)}}=1$, so only the second sink retains its
boundary term,
\[
  \bm{\chi}^\top\mathbf{L}\bm{\phi}
  = \sum_{v\in I}\chi_v\rho_v + (\mathbf{L}\bm{\phi})_{v_b^{(2)}} .
\]
On the right-hand side, $(\mathbf{L}\bm{\chi})_v=0$ on $I$ by
Eq.~\eqref{eq:chidef}, while $\phi_{v_b^{(1)}}=\phi_{v_b^{(2)}}=0$ removes both
boundary terms, so $\bm{\phi}^\top\mathbf{L}\bm{\chi}=0$.  The two sides differ
only because $\bm{\phi}$ vanishes on \emph{both} boundaries whereas $\bm{\chi}$
does not---this asymmetry in the Dirichlet data is precisely what the readout
exploits.  Equating them gives
$(\mathbf{L}\bm{\phi})_{v_b^{(2)}} = -\sum_{v\in I}\chi_v\rho_v$.  Since
$\mathbf{B}^\top\bm{J}=\mathbf{L}\bm{\phi}$, the divergence of $\bm{J}$ at
$v_b^{(2)}$ is $-\sum_{v\in I}\chi_v\rho_v$, i.e.\ the current \emph{into}
$v_b^{(2)}$ is $\sum_{v\in I}\chi_v\rho_v$.  Repeating the argument with
$\bm{1}-\bm{\chi}$, which is harmonic on $I$ by Eq.~\eqref{eq:chidef} and takes
the value $1$ at $v_b^{(1)}$ and $0$ at $v_b^{(2)}$, gives the second identity;
Eq.~\eqref{eq:pairing} follows by subtraction.
\end{proof}

Equation~\eqref{eq:pairing} is the whole readout.  By the discrete maximum
principle $0\leq\chi_v\leq1$ at every node, so each excited detector's vote
$1-2\chi_v$ lies in $[-1,+1]$, running from $+1$ when the detector sits on one
code boundary to $-1$ on the other: the decoder casts a convex vote of the
excited detectors along the learned harmonic coordinate.  No approximation and no free
parameter enters.  A single boundary node forfeits this entirely: it identifies
the two code boundaries as one element of $\partial G$, so no function
separating them exists, $\bm{\chi}$ is undefined, and there is nothing to pair
the syndrome with.

\begin{figure}[t]
\centering
\includegraphics[width=0.40\textwidth]{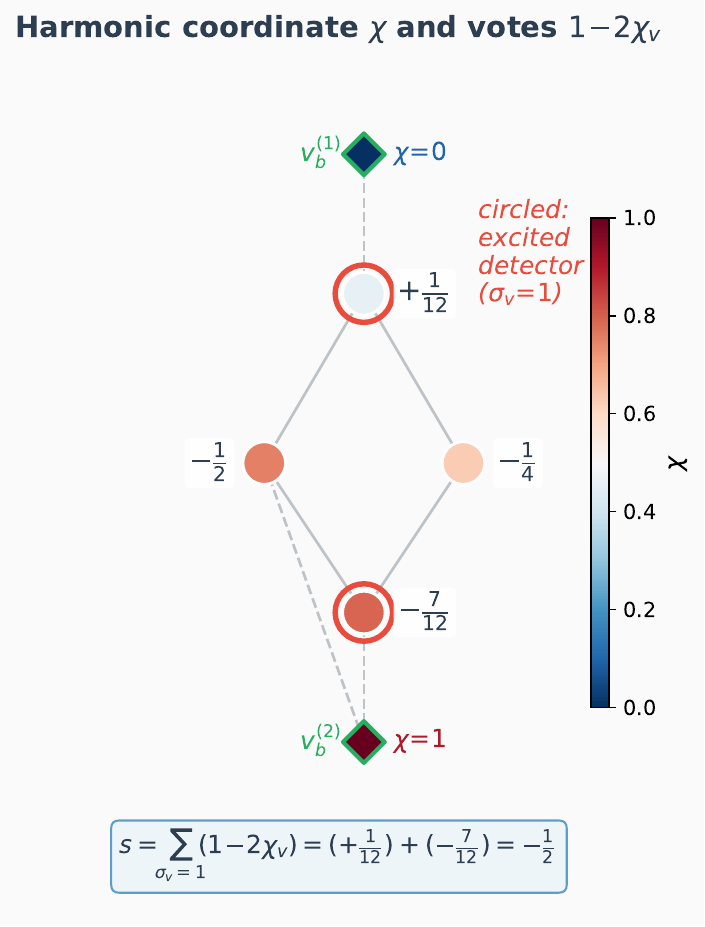}
\caption{%
  The pairing of Eq.~\eqref{eq:pairing} on the schematic graph of
  Fig.~\ref{fig:physics_current} (with $w=1$; the values shown are specific to
  this instance and change with the graph, the metric, and the syndrome).  Node
  colour is the harmonic coordinate
  $\chi$ ($0$ at $v_b^{(1)}$, $1$ at $v_b^{(2)}$, harmonic on the interior);
  each detector carries its vote $1-2\chi_v$, running from $+1$ at one code
  boundary to $-1$ at the other.  Summing the votes of the excited detectors
  (circled) gives $s=(+\tfrac{1}{12})+(-\tfrac{7}{12})=-\tfrac12$, the same
  scalar read as the sink-current difference $\Delta J$ in
  Fig.~\ref{fig:physics_current}(b).  
}
\label{fig:harmonic_vote}
\end{figure}

\subsection{Why a One-Dimensional Complex Makes the Readout Exact}
\label{sec:onedim}

It is worth making explicit why the current $\bm{J}$ admits the clean
two-term splitting of Eq.~\eqref{eq:hodge}, since it is a structural
consequence of the syndrome graph being a $1$-dimensional complex and not an
approximation we impose.  On a Riemannian domain the Hodge--Helmholtz
decomposition of any form has \emph{three} orthogonal parts,
\begin{equation}
  \bm{\omega} = \underbrace{d\bm{\alpha}}_{\text{exact}}
              \;\oplus\; \underbrace{\delta\bm{\beta}}_{\text{co-exact}}
              \;\oplus\; \underbrace{\bm{\gamma}}_{\text{harmonic}},
  \label{eq:hhk}
\end{equation}
where $d$ raises the form degree, $\delta$ (its metric adjoint) lowers it, and
$\bm{\gamma}$ is harmonic ($d\bm{\gamma}=\delta\bm{\gamma}=0$).  On a domain
\emph{with boundary} the three-way splitting requires boundary conditions to be
specified, and it is there that the distinction between absolute and relative
cohomology enters (Hodge--Morrey--Friedrichs).  This is not incidental to what
follows: the boundary is the whole content of our construction.

(An optional electromagnetic reading of these three terms---the Helmholtz
decomposition, and why the Aharonov--Bohm analogy is suggestive but ultimately
replaced by a grounded-resistor-network picture---is given in
Appendix~\ref{app:em}.)

A graph, however, is a simplicial complex of dimension one: it has $0$-cells
(nodes, carrying $0$-forms $\bm{\phi}$) and $1$-cells (edges, carrying $1$-forms
$\bm{J}$), but no $2$-cells.  With no $2$-form $\bm{\beta}$ available, the
co-exact sector of Eq.~\eqref{eq:hhk} is identically absent and the
decomposition of $\bm{J}$ collapses to the two terms of Eq.~\eqref{eq:hodge},
exact $\oplus$ harmonic: the graph current has no ``curl'' (solenoidal) part.


This is why the readout can be exact.  On a $1$-dimensional complex the only
obstruction to a clean statement would have been a co-exact sector, and there
is none; what remains is to say what the syndrome is measured against, which is
the content of Eq.~\eqref{eq:chidef} and Prop.~\ref{prop:pairing}.

The harmonic sector of the graph as a whole has dimension
$b_1 = |E|-|V|+1$, which counts every independent cycle of $G$.  A cycle of the
syndrome graph is a set of error mechanisms that triggers no detector, so most
of these cycles are local configurations with no logical content.  The logically
relevant subspace is smaller: it
is the homology \emph{relative to the boundary}, $H_1(G,\partial G;\Z_2)$,
i.e.\ cycles taken modulo those that can be closed through the code boundary.
Its dimension equals the number $k$ of logical qubits.

The Dirichlet condition has a twofold significance: the boundary nodes are
naturally read as an analytic device---a way of making the Poisson problem well
posed---while they also carry topological content.  \\  A
graph, as a combinatorial object, does not know which of its vertices form a
boundary; declaring the sinks \emph{is} the specification of $\partial G$.  This
matters because an error chain running from one code boundary to the other is
not a closed cycle, i.e. it has two endpoints, and therefore represents nothing in
the absolute homology $H_1(G)$.  It becomes a nontrivial cycle only in the
relative homology, where chains with endpoints on $\partial G$ are closed by
quotienting the boundary away (identifying all boundary nodes as a single point).  A single boundary node identifies the two code
boundaries as one element of $\partial G$, so a chain running between them is
relatively null-homologous and indistinguishable from a trivial loop; two
sinks keep them distinct and the chain acquires a nonzero class.  The same fact
appears analytically as the constant zero mode of $\mathbf{L}$: an unanchored
Laplacian is singular precisely because the boundary has not yet been specified,
and fixing it both removes the null mode and selects the homology.  The sinks
are therefore not an implementation detail but the boundary condition that
defines the topological object being measured.\\


\section{Method}
\label{sec:method}

\subsection{Overview}
\label{sec:overview}

Our decoder has three stages (Fig.~\ref{fig:architecture}):
\begin{enumerate}
  \item \textbf{GNN encoder}: maps the input syndrome $\bm{\sigma}$ to
    syndrome-adaptive edge weights $\bm{w}$.
  \item \textbf{Physics solver}: solves the discrete Gauss law
    $\mathbf{L}(\bm{w})\bm{\phi}=\bm{\rho}$ for the node potentials
    $\bm{\phi}$, then computes the edge current
    $\bm{J}=\mathrm{diag}(\bm{w})\mathbf{B}\bm{\phi}$.
  \item \textbf{Readout}: aggregates the edge-current vector $\bm{J}$ into a
    scalar logit for the logical error prediction.
\end{enumerate}

\begin{figure}[t]
\centering
\includegraphics[width=\textwidth]{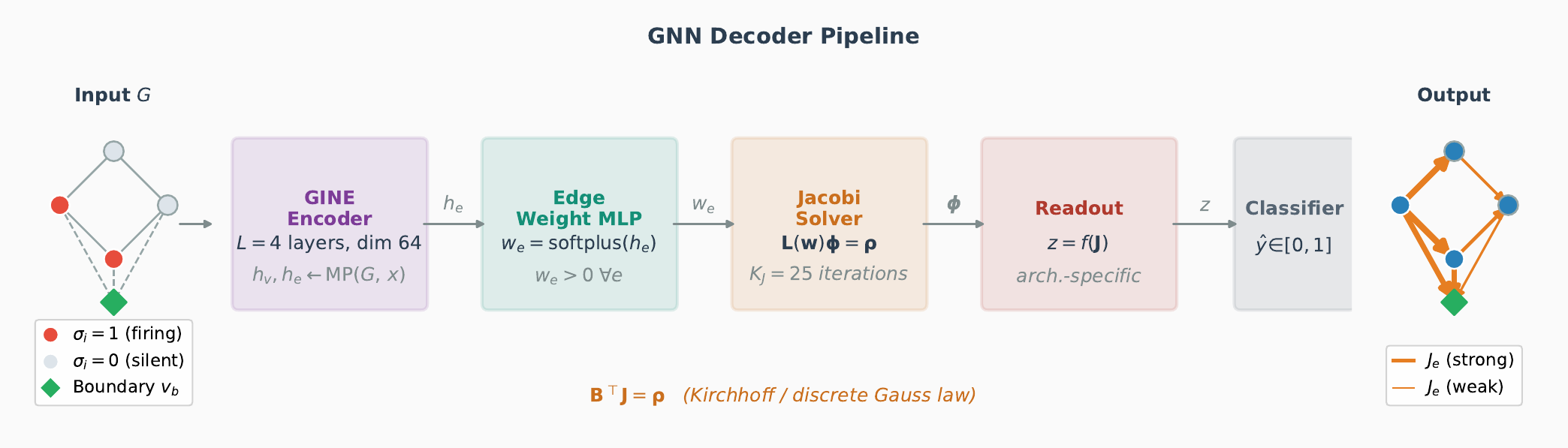}
\caption{%
  Pipeline of the physics-informed GNN decoder.
  The GINE encoder ($L=4$ message-passing layers, hidden dimension~64)
  maps the input syndrome graph $G$ to syndrome-adaptive edge embeddings,
  which an MLP converts to positive edge weights $w_e$.
  A Jacobi solver ($K_{\rm J}=25$ iterations) solves the discrete Poisson equation
  $\mathbf{L}(\bm{w})\bm{\phi}=\bm{\rho}$, yielding the node potentials
  $\bm{\phi}$ and the edge-current vector $J_e = w_e(\mathbf{B}\bm{\phi})_e$,
  which satisfies the discrete Gauss law $\mathbf{B}^\top\bm{J}=\bm{\rho}$
  exactly.  A readout (Sec.~\ref{sec:twosink}) reduces $\bm{\phi}$ or $\bm{J}$
  to a scalar, and a classifier MLP maps it to a logical-error
  probability $\hat{y}\in[0,1]$.
  The output panel illustrates schematic edge currents: arrow thickness
  is proportional to~$|J_e|$.
}
\label{fig:architecture}
\end{figure}

\subsection{Input Features}
\label{sec:features}

Each node $v\in V$ is assigned a feature vector
$\bm{x}_v\in\R^6$ encoding: syndrome bit value $\sigma_v\in\{0,1\}$,
boundary flag $\mathbbm{1}[v=v_b]$, spatial coordinates $(r_v, c_v)$,
and a stabiliser-type indicator (X/Z).  Each edge $e=(u,v)\in E$ has
features $\bm{a}_e\in\R^3$: a syndrome-product flag
$\sigma_u\cdot\sigma_v$, the MWPM log-likelihood cost $c_e$, and a
normalised spatial distance.  The source term is $\rho_v=\sigma_v$ for
internal nodes and $\rho_{v_b}=0$ for the boundary node.

\subsection{GNN Encoder}
\label{sec:encoder}

We use $L=4$ layers of Graph Isomorphism Network with Edge features
(GINE)~\cite{Hu2020} :
\begin{equation}
  \bm{h}_v^{(l+1)} = \mathrm{MLP}^{(l)}\!\left(
    (1+\epsilon)\bm{h}_v^{(l)}
    + \sum_{u\in\mathcal{N}(v)}\!\!\mathrm{ReLU}\!\left(
      \bm{h}_u^{(l)} + \bm{e}_{(u,v)}
    \right)
  \right),
\end{equation}
where $\bm{e}_{(u,v)}\in\R^h$ is a learnable projection of the edge features
and $h=64$ throughout and MLP, Multi-Layer Perceptron.  Node embeddings are initialised as
$\bm{h}_v^{(0)}=\mathrm{Linear}(\bm{x}_v)$.

\subsection{Edge Weights}
\label{sec:weights}

Given the final node embeddings $\{\bm{h}_v^{(3)}\}$, edge weights are
computed as:
\begin{equation}
  w_e = \exp\!\left(-c_e + \delta_e\right) + w_{\rm floor},
\label{eq:weights}
\end{equation}
where $c_e$ is the MWPM log-likelihood cost (from edge features) and
\begin{equation}
  \delta_e = \delta_{\max}\tanh\!\left(
    \frac{\tilde{\delta}_e}{\delta_{\rm temp}}
  \right), \qquad
  \tilde{\delta}_e = \mathrm{MLP}_{\rm head}\!\left(
    [\bm{h}_u; \bm{h}_v; \bm{a}_e]
  \right),
\label{eq:delta}
\end{equation}
is a syndrome-dependent correction to the MWPM weight.
The parameters $\delta_{\max}=2.0$, $\delta_{\rm temp}=2.0$, and
$w_{\rm floor}\in[10^{-5}, 10^{-3}]$ (decayed during training) enforce
numerical stability and prevent weight collapse.

The decomposition~\eqref{eq:weights} has a clear physical interpretation:
the base weight $e^{-c_e}$ recovers the edge weights that
MWPM~\cite{Edmonds1965,Dennis2002} assigns from the error channel, and
$\delta_e$ is a learned syndrome-adaptive correction to them.  When
$\delta_e\equiv 0$ the two decoders therefore share the same weights, and
differ only in how they use them: MWPM feeds them to a minimum-weight matching,
whereas our solver feeds them to the discrete Poisson equation~\eqref{eq:poisson}.
The learned correction $\delta_e$ is thus a deformation of the MWPM weights, not
of the MWPM algorithm.

\subsection{Poisson Solver (Jacobi)}
\label{sec:solver}

With weights $\bm{w}$ fixed, the decoder solves
\begin{equation}
  \mathbf{L}(\bm{w})\bm{\phi} = \bm{\rho}, \qquad
  \phi_{v} = 0 \ \ \text{for } v\in\partial G \;\;(\text{Dirichlet BC}),
\label{eq:poisson}
\end{equation}
using $K_{\rm J}=25$ steps of the Jacobi iterative method:
\begin{equation}
  \phi_v^{(k+1)} = \frac{\rho_v - \sum_{u\neq v} L_{vu}\phi_u^{(k)}}{L_{vv}},
  \qquad v\notin\partial G,
\end{equation}
with $\bm{\phi}^{(0)}=\bm{0}$.  In the single-sink architecture
$\partial G=\{v_b\}$; in the two-sink architecture
$\partial G=\{v_b^{(1)},v_b^{(2)}\}$ and both sinks are held at zero
(Sec.~\ref{sec:twosink}).  The update is differentiable with respect to
$\bm{w}$, allowing end-to-end training.  The Jacobi iteration is implemented
as a batched dense operation for efficiency on GPU.

\subsection{One sink J-Readout}
\label{sec:readout}

The edge-current vector $\bm{J}\in\R^{|E|}$ is computed via
Eq.~\eqref{eq:current}.  Each edge current $J_e$ is processed by a small
MLP to produce a per-edge feature $\bm{f}_e\in\R^{16}$.  A
gated-attention pooling operation~\cite{Li2015} aggregates
$\{\bm{f}_e\}_{e\in E}$ into a graph-level vector:
\begin{equation}
  \bm{z} = \mathrm{concat}\!\left[
    \bm{z}_{\rm gate},\; \overline{\bm{f}},\;
    \max_e\bm{f}_e,\; \mathrm{std}_e\bm{f}_e
  \right] \in \R^{64},
\end{equation}
which is passed to a two-layer MLP to produce the scalar logit.  We refer to
this single-boundary architecture as the \emph{single-sink} decoder; it uses
one virtual boundary node $v_b$ and reads the full edge-current field through
the gated-attention pool.

\subsection{Two-Sink J-Readout}
\label{sec:twosink}

As discussed in Sec.~\ref{sec:hodge}, the single-sink decoder drains all
syndrome charge into one boundary node, which merges the two code boundaries
linked by the logical operator.  No function separating them exists on that
graph, so there is no harmonic representative and no pairing to evaluate; the
gated-attention pool must instead learn a surrogate from the full current
field.

The two-sink readout resolves this by placing two boundary sinks, one on each
of the two boundaries connected by $\bar Z$, and evaluating the pairing
directly via Prop.~\ref{prop:pairing}.

\textbf{Boundary partition.}  Every boundary edge of the single-sink graph
(an edge from a detector to $v_b$, arising from a DEM error that lights a single
detector) is reassigned to one of two sinks $v_b^{(1)},v_b^{(2)}$ according to
the spatial coordinate of its detector endpoint along the logical axis.  The
logical axis is the code direction along which $\bar Z$ connects opposite
boundaries; we determine it empirically from the data (Sec.~\ref{sec:axis}) and
partition detectors by a threshold on that coordinate.

\textbf{Source and boundary conditions.}  We keep the syndrome as the Poisson
source, $\rho_v=\sigma_v$, so that the current is generated by the actual error
chain rather than by an externally imposed potential.  Both sinks are held at
the \emph{same} potential, $\phi_{v_b^{(1)}}=\phi_{v_b^{(2)}}=0$.

This is the hypothesis of Prop.~\ref{prop:pairing}, and it is not a gauge
choice.  It is the statement that the two code boundaries form a single
Dirichlet set while remaining distinct vertices, which is precisely the
quotient defining $H_1(G,\partial G;\Z_2)$: chains with endpoints on
$\partial G$ are closed by quotienting the boundary away, and a common
potential on $\partial G$ is that quotient written analytically.  Holding both
sinks fixed also renders $\mathbf{L}$ positive definite on the free nodes, so
the Jacobi iteration converges and no mean-subtraction of the source is
required; the drainage asymmetry is then a function of where the syndrome
charge sits relative to $\bm{\chi}$, and of nothing else.

Imposing a nonzero potential \emph{difference} across the two sinks instead (a
Laplace, source-free formulation) makes the current depend only on the graph
geometry and not on the syndrome, and does not train; we report this negative
control in Sec.~\ref{sec:results}.  Prop.~\ref{prop:pairing} explains that
control: in the source-free case $\bm{\phi}=\bm{\chi}$ exactly, and
$\bm{\chi}$ depends on the sample only through $\bm{w}(\bm{\sigma})$.

A remark on nomenclature is in order, since the construction may appear to be a
source--sink pair.  It is not: \emph{both} boundary nodes act as sinks.  The
source is the syndrome itself, distributed over the excited detectors, and
current drains from there into the two boundaries.  The two nodes play
physically equivalent roles; interchanging them maps $\bm{\chi}\mapsto
\bm{1}-\bm{\chi}$ and flips the sign of $s$, which is a relabelling of the two
logical classes and not a change of physics.  Which of the two boundaries
receives the larger net current is determined by the error configuration,
emerging sample by sample, and that is what encodes the logical class.  Had we
instead imposed a potential difference across the two boundaries, that sign
would be fixed by the boundary condition and would carry no information; this
is the negative control noted above.

\textbf{Harmonic readout.}  The logical signal is the net current between the
two sinks,
\begin{equation}
  s \;=\; \sum_{e\ni v_b^{(1)}} J_e \;-\; \sum_{e\ni v_b^{(2)}} J_e,
  \label{eq:twosink_signal}
\end{equation}
which by Prop.~\ref{prop:pairing} equals the pairing
$\sum_{v:\sigma_v=1}(1-2\chi_v)$ of the syndrome with the harmonic
representative of the relative cohomology generator.

Two points deserve emphasis, as they clarify what the second sink does and does
not do.  First, it does not remove the exact part of the current from the
field.  The field is exact in its entirety (Eq.~\eqref{eq:jharm_zero}) and
there is nothing to remove.  What the second sink supplies is the boundary
data without which the relative class, and hence the harmonic representative
$\bm{\chi}$ that measures it, does not exist at all.  The distinction matters:
the second sink is not a filter applied to a field, it is the datum that makes
the field's logical content definable.

Second, $\bm{\chi}$ never has to be constructed at inference.  By
Prop.~\ref{prop:pairing} the pairing is already reported by the currents
incident on the two sinks, which is why the readout is a single scalar and
carries no learned parameters.  That the same quantity can be evaluated
directly from $\bm{\chi}$ gives an independent check on the implementation, and
is the explicit harmonic readout we return to in Sec.~\ref{sec:future}.

A small MLP maps the scalar $s$ to the logit.  This is a calibration of a single
pre-computed scalar, not a learned readout: the pairing itself is evaluated in
closed form by Prop.~\ref{prop:pairing}, with no parameters.  The calibration
absorbs the irreducible ambiguity of near-degenerate (Nishimori) error
configurations that even MWPM misdecodes; at larger code distance, where the two
boundaries are farther apart and the current field is larger and sparser, the
pairing is a noisier quantity to read (Sec.~\ref{sec:scaling}), and the
calibration correspondingly does more of the work --- a contingent property of
the signal-to-noise, not of the construction, which remains exact at every
distance.  In contrast to the single-sink gated-attention pool,
which must learn a surrogate for the pairing from the full $|E|$-dimensional
current field, the two-sink readout evaluates the pairing in closed form and
leaves only the calibration to training.

\textbf{Why not more than two sinks.} Proposition~\ref{prop:pairing} shows that
two sinks suffice for $k=1$; it is worth noting that more would not help.  Two
sinks are not a minimal choice one could refine by adding more: they are the
exact choice.  Partitioning the boundary detectors into four groups on a
single-logical-qubit code does not increase resolution: two sub-sinks lying on
the \emph{same} logical boundary belong to the same relative class, so no
harmonic coordinate separates them and their pairing with the syndrome is
identically zero; they contribute no logical signal while adding currents the
calibration must learn to ignore.  More sinks than $2k$ therefore impose a
richer boundary structure than the code possesses, degrading rather than
improving the readout.  Four sinks become the correct construction only when
$k=2$ on a code that possesses four boundary components---for instance a planar
code with holes, or two lattice-surgery patches---with two independent boundary
pairs, one per logical qubit.  Closed-surface codes such as the toric code fall
outside the construction as stated: having no boundary at all, they admit no
boundary sinks, and their $k=2g$ logical operators would have to be read against
non-contractible cycles rather than against a relative class.  The number of
readout points is thus not a design choice but is dictated by the topology.

\subsection{Alternative Readouts (Ablation)}
\label{sec:alt_readouts}

To isolate the role of the physics solver and the readout choice, we implement
two further variants on the single-sink architecture, keeping the GNN encoder,
edge-weight parametrisation~\eqref{eq:weights}, and readout MLP \emph{identical}
to the single-sink $J$-readout baseline.  These serve as a controlled ablation
of the physics content of the readout (Sec.~\ref{sec:readout_ablation}).

\textbf{$\phi$-readout}: The Poisson equation~\eqref{eq:poisson} is solved
with the same damped Jacobi scheme.  However, the logit is computed from the
node potentials $\bm{\phi}$ (pooled over the $N-1$ internal nodes via a
gated-attention mechanism) rather than from the edge currents $\bm{J}$.
As we discuss in Sec.~\ref{sec:readout_interp}, the node-potential vector is
a 0-form on the syndrome graph: it carries no orientation information and
does not encode the topological flux directly.

\textbf{Diffusion-readout}: Instead of solving the steady-state Poisson
equation, we evolve the initial density $\bm{\rho}$ under the diffusion
operator $e^{-t\mathbf{L}(\bm{w})}$ for a total time $t$, approximated by
$K_{\rm D}=15$ forward-Euler steps:
\begin{equation}
  \bm{u}^{(k+1)} = \bm{u}^{(k)} - \Delta t\,\mathbf{L}(\bm{w})\bm{u}^{(k)},
  \qquad \bm{u}^{(0)} = \bm{\rho},
\label{eq:diffusion}
\end{equation}
where $\Delta t$ is a learnable scalar constrained to $(0,\Delta t_{\max})$ via
a sigmoid parametrisation.  The readout uses the diffused current
$\bm{J}_{\rm diff}=\mathrm{diag}(\bm{w})\mathbf{B}\bm{u}$.
Unlike the Poisson current~\eqref{eq:current}, $\bm{J}_{\rm diff}$ does
\emph{not} satisfy the discrete Gauss law at finite~$t$:
$\mathbf{B}^\top\bm{J}_{\rm diff} = \mathbf{L}(\bm{w})\bm{u}(t) \neq \bm{\rho}$
(equality holds only as $t\to\infty$, recovering the Poisson solution).
The physical significance of this violation is discussed in
Sec.~\ref{sec:readout_interp}.

\subsection{Loss Function}
\label{sec:loss}

The total training loss is
\begin{equation}
  \mathcal{L} = \mathcal{L}_{\rm BCE}
              + \lambda_{\rm res}\,\mathcal{L}_{\rm res}
              + \lambda_{\delta}\,\mathcal{L}_{\delta},
\label{eq:loss}
\end{equation}
where $\mathcal{L}_{\rm BCE}=\mathrm{BCE}(\sigma(\text{logit}), y)$ is the
binary cross-entropy with class-balanced positive weight,
$\mathcal{L}_{\rm res}=\sum_{v\notin\partial G}
\bigl[(\mathbf{L}\bm{\phi})_v-\rho_v\bigr]^2$ is a
physics-residual penalty (active only for $J$- and $\phi$-readouts).  The
restriction to $v\notin\partial G$ is not cosmetic: the Dirichlet nodes do not
satisfy the Gauss law by construction---they absorb the current---so including
them would penalise $\sum_{e\ni v_b^{(1)}}J_e$ and $\sum_{e\ni v_b^{(2)}}J_e$,
that is, the readout signal itself, driving $s\to0$.  Finally,
$\mathcal{L}_{\delta}=\norm{\tilde{\bm{\delta}}}^2$ regularises the raw
weight corrections to prevent tanh saturation.  Weights are set to
$\lambda_{\rm res}=1.0$, $\lambda_{\delta}=10^{-3}$ for $J$/$\phi$, and
$\lambda_{\rm res}=0$ for the diffusion model.

\section{Experiments}
\label{sec:experiments}

\subsection{Dataset}
\label{sec:dataset}

We use the distance-$5$ \stimcirc{}
circuit with $r=5$ syndrome rounds, giving a syndrome graph of $N=121$ nodes
(120 space-time detectors plus one virtual boundary).  Syndromes are generated
with \textsc{Stim}~\cite{Gidney2021} under circuit-level depolarising noise at
$p\in\{0.005,0.010,0.015\}$.  Each sample stores the syndrome graph with
node/edge features, the MWPM prediction $\hat{y}_{\rm MWPM}$, and the
ground-truth logical-error label $y\in\{0,1\}$.  For each $(p,\text{seed})$ the
data are split stratified by class into 80\% train / 10\% validation / 10\%
test; we use five seeds, each generating its own split, so that architecture
comparisons are paired on the same split and reported as statistics over seeds.

\subsection{Training}
\label{sec:training}

All models use the Adam optimiser~\cite{Kingma2015} with base learning rate
$10^{-3}$ and batch size 32.  The learning-rate schedule is a linear warmup over
the first 20 epochs followed by cosine decay to $10^{-5}$; the warmup removes
the dead-start failures we observed with a bare cosine schedule at higher $p$.
Validation balanced accuracy is evaluated every epoch, and we apply early
stopping with a patience of 60 epochs (maximum 600 epochs).  A floor-decay
schedule reduces $w_{\rm floor}$ from $10^{-1}$ to $10^{-3}$ over training
(Sec.~\ref{sec:floor}).  Model selection uses the checkpoint of maximum
validation balanced accuracy.

\textbf{Floor-decay guardrail.}
\label{sec:floor}
Premature reduction of $w_{\rm floor}$ can destabilise the Jacobi solver;
we anneal it geometrically over the schedule so that the Poisson residual
$\|\mathbf{L}\bm{\phi}-\bm{\rho}\|/\|\bm{\rho}\|$ stays small throughout
training.

\subsection{Evaluation Metric}
\label{sec:metric}

We report balanced accuracy ${\rm BA} = \frac{1}{2}({\rm TPR}+{\rm TNR})$ at the
validation-selected threshold, and the threshold-independent Area Under the Curve (AUC).  Because the
decision threshold selected on a noisy validation set is itself a source of
seed-to-seed variance at high $p$, the AUC is the more robust metric there and
we quote both.  Architecture comparisons use the paired (same-split) difference
$\Delta\mathrm{BA}$ across seeds.

\section{Results}
\label{sec:results}

\subsection{Logical Axis Determination}
\label{sec:axis}

The two-sink readout requires knowing which code direction the logical
operator $\bar Z$ connects.  We determine this empirically, without appeal to
the lattice convention (which varies with how \textsc{Stim} emits detector
coordinates at a given distance).  For each spatial axis we partition the
boundary detectors by the median coordinate and test how well the purely
topological rule ``the error chain touches boundary detectors on both sides''
predicts the true logical label $y$.  The axis for which this rule correlates
with $y$ is the logical axis; the other axis is uncorrelated (near-random).
This zero-training test also yields a useful diagnostic: the balanced accuracy
of the bare topological rule measures how much of the logical signal is already
present in the completion structure of the syndrome, before any weights are
learned.

\subsection{Two-Sink vs.\ Single-Sink}
\label{sec:twosink_results}

The controlled comparison between the single-sink and two-sink decoders is
reported in Table~\ref{tab:readout}, together with the other two readouts.  The
architectures share the same GINE encoder,
edge-weight parametrisation, Jacobi solver, loss, optimiser, learning-rate
schedule (linear warmup followed by cosine decay), early-stopping criterion,
and per-seed data split; they differ only in the boundary structure and the
readout (full-field gated attention for single-sink, harmonic two-sink signal
of Eq.~\eqref{eq:twosink_signal} for two-sink).  We report the balanced
accuracy (BA) at the validation-selected threshold and the threshold-independent
AUC, each as mean\,$\pm$\,standard deviation over five seeds, together with the
MWPM reference on the same code and noise.

\textbf{The two-sink readout is statistically equivalent to both the
single-sink current readout and the best single-sink readout ($\phi$).}
At $p=0.005$ the two-sink decoder reaches $\mathrm{BA}=0.898\pm0.013$
(Table~\ref{tab:readout}), between the single-sink $J$-readout's
$0.893\pm0.009$ and the $\phi$-readout's $0.910\pm0.018$.  Neither paired
difference is resolved at five seeds: two-sink minus single-sink is
$+0.005\pm0.012$ ($t=0.9$, $p=0.43$), and two-sink minus $\phi$ is
$-0.012\pm0.024$ ($t=-1.1$, $p=0.33$).  We therefore do \emph{not} claim that
the two-sink decoder is the most accurate, nor that it improves on the
single-sink pool: within seed variance it ties both.  As discussed in
Sec.~\ref{sec:readout_ablation}, this parity is the informative result: a single
topological scalar, carrying no learned readout parameters, saturates the same
logical information that the full-field readouts access, exactly as expected
when the logical signal is one-dimensional ($k=1$).  The value of the harmonic
readout is interpretability, parameter economy, and principled generalisation
to $k>1$, not raw accuracy.

\textbf{Behaviour as $p$ grows.}
The quantitative comparison of Table~\ref{tab:readout} is at $p=0.005$, the
regime in which the distance-$5$ logical signal is well resolved.  As $p$
increases the signal weakens and, past the distance-$5$ threshold, all decoders
including MWPM degrade towards the random baseline; in that regime the code
protects too little for a meaningful comparison among readouts, and we do not
report one.  Characterising the high-$p$ regime---which on larger codes requires
a dedicated treatment of the weak-signal limit---is left to future work.

\textbf{Negative control.}
As a control, we drive the two-sink current with an imposed potential difference
rather than the syndrome source.  In this Laplace formulation the syndrome
enters only through the weights, and the decoder does not train
($\mathrm{AUC}=0.50$): with uniform weights the current is the same for every
sample, so no gradient reaches the encoder.  The logical signal must therefore
enter through the source $\rho_v=\sigma_v$, not through the boundary conditions.
Proposition~\ref{prop:pairing} predicts exactly this: in the source-free case
$\bm{\phi}=\bm{\chi}$, which is sample-independent.

\subsection{Readout Comparison}
\label{sec:readout_ablation}

Table~\ref{tab:readout} compares all four readouts at fixed noise
($d=5$, $p=0.005$, five seeds, identical 600-epoch protocol with warmup and
early stopping).  Two of them read the edge current $\bm{J}$ and differ only in
boundary topology: the single-sink $J$-readout pools the current at one
boundary node, while the two-sink $J$-readout reads the net current between the
two sinks, which---by the construction of Sec.~\ref{sec:twosink}---is the
harmonic (topological) projection.  The other two read different quantities:
the $\phi$-readout (node potentials) and the diffusion readout.

\begin{table}[t]
\centering
\caption{%
  Readout comparison at $d=5$, $p=0.005$, five seeds, identical protocol.
  BA at the validation-selected threshold and threshold-independent AUC, as
  mean\,$\pm$\,sample std.  The two-sink, $\phi$, and single-sink $J$ readouts
  are statistically equivalent at five seeds; all three beat the diffusion
  readout.  The last row is the minimum-weight perfect matching reference on the
  same DEM, which none of the learned decoders reaches.
}
\label{tab:readout}
\begin{tabular}{lcc}
\toprule
Readout & BA & AUC \\
\midrule
$J$-readout (two-sink)   & $0.898\pm0.013$ & $0.958\pm0.010$ \\
$\phi$-readout           & $0.910\pm0.018$ & $0.974\pm0.006$ \\
$J$-readout (single-sink) & $0.893\pm0.009$ & $0.957\pm0.007$ \\
Diffusion-readout        & $0.875\pm0.023$ & $0.953\pm0.021$ \\
\midrule
MWPM (reference)         & $0.980$ & --- \\
\bottomrule
\end{tabular}
\end{table}

Two features stand out, and both revise the expectation set by the original
DEC argument.  First, \emph{there is no readout hierarchy of the predicted
form}.  The naive DEC reasoning suggested $J>\phi>\text{diffusion}$, with the
$1$-form current superior to the $0$-form potential and the non-conserved
diffusion current failing outright.  The clean multi-seed runs show a different
ordering: the $\phi$-readout ties the two-sink readout for best, the diffusion
readout trains successfully (Sec.~\ref{sec:readout_interp}) rather than
collapsing, and the single-sink $J$-readout is only
third.  The catastrophic diffusion failure reported in preliminary experiments
was an artefact of training without warmup, not an intrinsic gradient pathology.

Second, and more informative, \emph{the two best readouts (two-sink and
$\phi$) are statistically equivalent} (paired difference
$\Delta\mathrm{BA}=-0.012\pm0.024$, $t=-1.1$, over the five same-split seeds).  This
parity is the central finding, and it is exactly what the Hodge picture
predicts.  The $\phi$-readout extracts the logical signal empirically from the
full node-potential field with a pooling MLP; the two-sink readout evaluates it
in closed form as the pairing $\langle\bm{\chi},\bm{\sigma}\rangle$
(Eq.~\eqref{eq:twosink_signal}), a single scalar.  That a one-dimensional
topological readout matches a full-field readout is not a coincidence but a
consequence of the logical signal being one-dimensional
($k=\dim H_1(G,\partial G;\Z_2)=1$, Sec.~\ref{sec:onedim}): a single harmonic
scalar already saturates the
information the field readout accesses.

The value of the harmonic readout is therefore not higher accuracy but
\emph{interpretability and principled generalisation}.  It equals the best
empirical readout while using no learned readout parameters; it states
\emph{why} it works (it evaluates a topological invariant); and it generalises
in a controlled way, one per boundary pair, whereas the field readouts have no comparable
structure.  The contrast with the single-sink $J$-readout sharpens the point:
reading the current field alone gives $0.893$, whereas evaluating the
pairing through the two-sink boundary structure gives $0.898$: a difference of
$+0.005\pm0.012$ ($t=0.9$) that does not survive seed variance.  The two
readouts are equivalent in accuracy; what the boundary structure buys is not a
higher number but a closed-form, parameter-free evaluation of the same logical
pairing the single-sink pool must learn to approximate.

\subsection{Scaling to Larger Distance}
\label{sec:scaling}

The $d=5$ comparison shows the two-sink and $\phi$ readouts on a par.  A natural
question is whether this picture persists at larger code distance, where the
logical signal occupies a larger and sparser current field.  We therefore repeat
the comparison at $d=7$, $p=0.005$, five seeds, under the same protocol (with a
longer patience and a minimum epoch count, since the larger graph converges more
slowly).  Table~\ref{tab:d7} reports the results.

\begin{table}[t]
\centering
\caption{%
  Readout comparison at $d=7$, $p=0.005$, five seeds, all runs converged.
  BA at the validation-selected threshold and threshold-independent AUC, as
  mean\,$\pm$\,sample std.  The two-sink readout significantly exceeds the
  single-sink $J$-readout (paired $t=4.1$, $p=0.014$) and ties the
  $\phi$-readout (paired $t=1.6$, not significant; see text on the one $\phi$
  dead-start seed).  The last row is the minimum-weight perfect matching
  reference on the same DEM.
}
\label{tab:d7}
\begin{tabular}{lcc}
\toprule
Readout & BA & AUC \\
\midrule
$J$-readout (two-sink) & $0.857\pm0.011$ & $0.935\pm0.006$ \\
$\phi$-readout      & $0.815\pm0.061$ & $0.895\pm0.030$ \\
$J$-readout (single-sink) & $0.811\pm0.021$ & $0.885\pm0.017$ \\
Diffusion           & $0.760\pm0.080$ & $0.860\pm0.041$ \\
\midrule
MWPM (reference)    & $0.989$ & --- \\
\bottomrule
\end{tabular}
\end{table}

At $d=7$ the two-sink readout reaches $\mathrm{BA}=0.857\pm0.011$, and the
comparison sharpens relative to $d=5$ in one direction while holding in the
other.  Against the single-sink $J$-readout ($0.811\pm0.021$) the paired
same-split difference is $+0.046\pm0.025$ ($t=4.1$, $p=0.014$; positive on all
five seeds, Wilcoxon $p=0.06$): a genuine and significant advantage, and a
stronger result than the parity seen at $d=5$.  Isolating the harmonic pairing
through the boundary structure thus helps more, not less, as the current field
grows larger and sparser---consistent with a full-field pool having a harder
job at larger distance.

Against the $\phi$-readout ($0.815\pm0.061$) the two-sink readout remains
statistically equivalent: the paired difference is $+0.042\pm0.057$ but
$t=1.6$ ($p=0.18$), and it is not robust.  Almost all of the nominal gap comes
from a single seed on which the $\phi$-readout hit a dead-start
($\mathrm{BA}=0.713$, against $0.82$--$0.86$ on the other four); excluding that
seed the difference falls to $+0.018$ and remains non-significant, and a
Wilcoxon signed-rank test over all five seeds likewise does not resolve it
($p=0.13$).  We therefore read the two-sink and $\phi$ readouts as on a par at
$d=7$ as at $d=5$, and do \emph{not} claim an advantage over $\phi$: the
higher two-sink \emph{mean} reflects a $\phi$ baseline failure on one seed, not
a systematic separation.  The defensible $d=7$ statement is thus that the
closed-form scalar significantly exceeds the single-sink pool and ties the best
full-field readout, exactly as expected when the logical signal remains
one-dimensional.

Two observations support the interpretation.  First, the bare topological signal
(the harmonic readout at uniform weights, $w=1$, evaluated from the corrected
harmonic coordinate $\bm{\chi}$) drops from $0.63$ at $d=5$ to $0.56$ at $d=7$:
the harmonic coordinate $\bm{\chi}$ varies over more detectors at $d=7$, so the
per-detector votes $1-2\chi_v$ are individually smaller and the syndrome's
signal-to-noise against $\bm{\chi}$ falls.  That the trained decoder nonetheless
reaches $0.857$ shows that the learned metric $\bm{w}(\bm{\sigma})$
reconcentrates the votes onto the relevant detectors; the weak bare signal makes
the learned metric more necessary at larger distance, not less.  Second, we verified that the boundary
partition is optimal: sweeping the split threshold yields a symmetric $72/72$
partition of the boundary detectors at the operating point, with the bare BA
flat across all balanced partitions, the signal is genuinely weak, not an
artefact of a mis-placed cut.

\section{Discussion}
\label{sec:discussion}

\subsection{DEC Interpretation}
\label{sec:dec_interp}

The central insight of our formulation is that the logical error signal
$\bm{J}=\mathrm{diag}(\bm{w})\mathbf{B}\bm{\phi}$ is a \emph{discrete
1-form} (a co-chain) on the syndrome graph, not a scalar or a node
embedding.  What carries the logical-error parity is not a component of
$\bm{J}$---the solver's current is a pure gradient flow (Eq.~\eqref{eq:jharm_zero})---
but the pairing of the syndrome with the harmonic coordinate $\bm{\chi}$ fixed
by the two boundary sinks.  This pairing is an exact discrete quantity, defined
by the boundary data alone and carrying no free parameters
(Prop.~\ref{prop:pairing}).

The connection to MWPM is transparent: setting $\delta_e\equiv 0$
recovers the MWPM weights, on which the solver then runs its discrete Poisson
readout in place of a matching.  The GNN learns syndrome-adaptive corrections
$\bm{\delta}(\bm{\sigma})$ on top of these weights, in principle capturing
correlations that a fixed-weight matching ignores.

\textbf{Generalisation to $k$ logical qubits.}
The two boundary sinks in our architecture are not an arbitrary implementation
choice: they encode the fact that the planar rotated surface code has exactly
$k=1$ logical qubit, read against its two code boundaries.  The connection can
be made precise via the DEC framework.

On a connected graph with no boundary conditions, the weighted Laplacian
$\mathbf{L}(\bm{w})$ is positive semi-definite with a one-dimensional null
space spanned by the constant vector $\mathbf{1}$.  Fixing the two sinks removes
this null mode and, more importantly, specifies $\partial G$: it is the datum
that turns the error chain into a relative cycle and fixes the harmonic
coordinate $\bm{\chi}$ against which the syndrome is paired
(Sec.~\ref{sec:hodge}).  For $k=1$ there is one such coordinate, and the pairing
of Prop.~\ref{prop:pairing} encodes the logical-error class.

For a code encoding $k>1$ logical qubits, the same reasoning applies with $k$
independent boundary-sink \emph{pairs}, i.e.\ $2k$ sinks in total.  Each logical
qubit $i$ is associated with a pair of topologically distinct code boundaries;
declaring that pair fixes a harmonic coordinate $\bm{\chi}^{(i)}$, and
Prop.~\ref{prop:pairing} applied to it returns the pairing
$\langle\bm{\chi}^{(i)},\bm{\sigma}\rangle$, which feeds an independent output
head predicting the $i$-th logical error.  The architecture extends by adding
$2k$ sinks and $k$ readout heads with no other structural change, provided the
code has $2k$ boundary components against which the pairs can be defined.  This
proviso is not vacuous: closed-surface codes such as the toric code have no
boundary, and their logical operators are read against non-contractible cycles
rather than against a relative class, which the present construction does not
cover.

This generalisation is transparent within the DEC language: the boundary sink
pairs correspond to the generators of the relative homology group
$H_1(G,\partial G;\mathbb{Z}_2)$ of the syndrome graph, and the $k$ pairings
measure the coupling of the error chain with each generator.  The two-sink
readout of Sec.~\ref{sec:twosink} is precisely the $k=1$ instance realised
correctly: the two sinks resolve the single logical cycle into a pair of
boundary sinks, so that the net current between them evaluates the pairing of
Eq.~\eqref{eq:pairing}.  The single-sink decoder collapses this pair into one
node, leaving no separating harmonic coordinate to pair against, which is why,
at low $p$ where the logical signal is well resolved, the two-sink readout
improves on it (Sec.~\ref{sec:twosink_results}).

\subsection{Physical Interpretation of the Readout Comparison}
\label{sec:readout_interp}

The multi-seed results (Sec.~\ref{sec:readout_ablation}) show the two-sink and
$\phi$ readouts on a par, the diffusion readout trailing but training
successfully, and the single-sink $J$-readout weakest.  This subsection explains
the mechanisms behind that ordering---why conservation lets the two-sink readout
expose the topological signal directly, why the $\phi$-readout accesses the same
information less directly, and why the diffusion readout suffers a
vanishing-gradient handicap that nonetheless proves non-fatal.  We had initially
expected these mechanisms to produce a strict hierarchy $J>\phi>$~diffusion; the
data do not bear that out, and the analysis below is therefore diagnostic---it
identifies which effects matter and which we had overstated---rather than
predictive.

\textbf{Kirchhoff conservation and topological content.}
The Poisson edge current $\bm{J}=\mathrm{diag}(\bm{w})\mathbf{B}\bm{\phi}$
satisfies the discrete Gauss law
\begin{equation}
  \mathbf{B}^\top \bm{J} = \mathbf{L}(\bm{w})\bm{\phi} = \bm{\rho},
\label{eq:kirchhoff}
\end{equation}
i.e., Kirchhoff's current law: the net current flowing into each internal
node equals the local source $\rho_v$.  Conservation is what lets the two-sink
readout be read at the boundary rather than in the bulk: because $\bm{J}$ is
divergence-free away from the sinks, the difference of the currents they drain
equals the pairing $\sum_{\sigma_v=1}(1-2\chi_v)$ of Prop.~\ref{prop:pairing},
independently of any particular set of edges.  (The flux across a cut, by
contrast, returns the enclosed syndrome charge, a gradient quantity that does
\emph{not} encode the logical class; this is why the readout is a sink-current
difference and not a cut integral, Fig.~\ref{fig:physics_current}.)

The $\phi$-readout does not directly access this topological invariant.
The node potentials $\bm{\phi}$ are a discrete 0-form: they encode scalar
values at nodes, with no intrinsic orientation.  Even though $\bm{\phi}$
and $\bm{J}$ carry the same information in principle (one determines the
other via $\bm{J}=\mathrm{diag}(\bm{w})\mathbf{B}\bm{\phi}$), the
pooled scalar field $\bm{\phi}$ exposes it less directly to the readout
MLP: the logical-class distinction requires the readout to implicitly
reconstruct path-oriented information from a smooth scalar field,
which is a harder function to learn.

The diffusion current $\bm{J}_{\rm diff}=\mathrm{diag}(\bm{w})\mathbf{B}\bm{u}(t)$
at finite~$t$ violates~\eqref{eq:kirchhoff}:
$\mathbf{B}^\top\bm{J}_{\rm diff} = \mathbf{L}(\bm{w})\bm{u}(t) \neq \bm{\rho}$.
The diffused field $\bm{u}(t)$ does not satisfy Poisson balance, so the
current it induces is \emph{not conserved}.  In a non-conserved current
field there is no well-defined net flux between boundaries, and the
topological-class signal is absent.  The Poisson steady state can be
recovered by taking $t\to\infty$, but this limit is inaccessible with a
finite number of Euler steps.

\textbf{Gradient path to the edge weights.}
All three readouts are trained with the same binary cross-entropy loss
$\mathcal{L}_{\rm BCE}$ and the same physics-residual regularisation.
The key difference lies in how the loss gradient reaches the edge weights~$\bm{w}$.

For the $J$-readout, the gradient $\partial J_e / \partial w_e$ has an
explicit, direct component:
\begin{equation}
  \frac{\partial J_e}{\partial w_e}
  = \frac{\partial}{\partial w_e}\bigl[w_e(\mathbf{B}\bm{\phi})_e\bigr]
  = (\mathbf{B}\bm{\phi})_e + w_e\,\frac{\partial(\mathbf{B}\bm{\phi})_e}{\partial w_e}.
\label{eq:grad_J}
\end{equation}
The first term, $(\mathbf{B}\bm{\phi})_e = \phi_u - \phi_v$ (the potential
drop across edge~$e$), provides a direct gradient signal that does not pass
through the Jacobi solver.  As long as a syndrome is present ($\bm{\rho}\neq\bm{0}$),
this term is generically nonzero, giving a strong first-order signal.

For the $\phi$-readout, the logit depends on $w_e$ only through the
Poisson solution $\bm{\phi} = \mathbf{L}(\bm{w})^{-1}\bm{\rho}$.
The implicit-function theorem gives
$\partial\bm{\phi}/\partial w_e = -\mathbf{L}(\bm{w})^{-1}(\partial\mathbf{L}/\partial w_e)\bm{\phi}$,
a valid but indirect gradient that requires the full Laplacian inverse.
With only $K_{\rm J}=25$ Jacobi iterations the unrolled computation is an
approximation to this inverse; the resulting truncated gradient is weaker
and can miss the relevant directions.  More importantly, this is the
\emph{only} gradient path to~$\bm{w}$ for the $\phi$-readout, since
the readout does not depend on $\bm{w}$ explicitly.

For the diffusion-readout, the gradient passes through $K_{\rm D}$ matrix
multiplications $(I - \Delta t\,\mathbf{L})^{K_{\rm D}}$, each of which attenuates
eigenmodes by a factor $(1-\Delta t\,\lambda_i)$.  For eigenvalues
$\lambda_i \sim 1$ and $\Delta t\sim 0.1$--$0.5$, the $K_{\rm D}$-step product
$(1-\Delta t\,\lambda_i)^{K_{\rm D}}$ can be as small as $10^{-5}$ for $K_{\rm D}=15$,
leading to severe vanishing gradients analogous to those in unrolled RNNs.
The direct term in~\eqref{eq:grad_J} is still present (since the
diffusion-readout also computes $J=\mathrm{diag}(\bm{w})\mathbf{B}\bm{u}$),
but $\bm{u}(t)$ after $K_{\rm D}$ steps is a heavily smoothed version of $\bm{\rho}$
with attenuated syndrome structure, reducing the magnitude of
$(\mathbf{B}\bm{u})_e$.
This vanishing-gradient effect is real but, contrary to our initial
expectation, \emph{not fatal} under a well-tuned training protocol.  With a
linear learning-rate warmup and a sufficient epoch budget, the diffusion
readout escapes the flat initial region and trains to
$\mathrm{BA}=0.875\pm0.023$ (Table~\ref{tab:readout}), below the two best
readouts but far above random and above the single-sink $J$-readout.  The
catastrophic ``frozen at $\ln 2$'' behaviour seen in preliminary runs was an
artefact of training without warmup and with too few epochs, not an intrinsic
obstruction.  The direct term in~\eqref{eq:grad_J} evidently supplies enough
signal, once the optimiser is past the warmup, to overcome the attenuated
implicit path.

\subsection{Limitations and Future Work}
\label{sec:future}

The primary limitation is that the decoder does not outperform MWPM under
circuit-level depolarising noise with known $p$.  This is expected: Dennis
et al.~\cite{Dennis2002} established that MWPM on the Nishimori line is
near-optimal for this noise model, and a data-driven decoder cannot
systematically improve on maximum-likelihood inference.  A second limitation is that the comparison is meaningful only where the logical
signal is well resolved: as $p$ increases past the distance-$5$ threshold all
decoders, MWPM included, degrade towards the random baseline, and we make no
claim of an ordering among readouts in that regime.

We identify three directions to extend the regime of usefulness:

\textbf{(i) Larger code distance.}
The distance-$5$ code protects the logical signal only at low $p$.  A larger
distance ($d=7,9$) extends the regime in which the signal is well resolved, and
is the natural test of whether the two-sink advantage persists to higher $p$ on
codes large enough to express it.  Because the two-sink construction is defined homologically
(one boundary-sink pair per generator of the relative homology), it
transfers to larger $d$
without structural change.

\textbf{(ii) Explicit harmonic readout and $k>1$ codes.}
The two-sink signal of Eq.~\eqref{eq:twosink_signal} evaluates the pairing
implicitly, through the currents incident on the sinks.  Solving
Eq.~\eqref{eq:chidef} for $\bm{\chi}$ directly and forming
$\sum_{v:\sigma_v=1}(1-2\chi_v)$ evaluates it explicitly; by
Prop.~\ref{prop:pairing} the two must agree, which makes this a sharp check on
the implementation rather than a variant of it.  It also raises a question we
leave open: since the pairing depends on the syndrome only through which
detectors are excited, and on the metric only through $\bm{\chi}$, a decoder
built on Eq.~\eqref{eq:chidef} alone would dispense with the Poisson solve for
$\bm{\phi}$ at inference.  The same route generalises to codes with $k>1$
logical qubits, where one reads $k$ independent pairings, one per boundary
pair.

\textbf{(iii) Bias-aware training.}
Under biased noise the correct metric is anisotropic.  Initialising the weights
from the biased DEM costs ($\log w_e^{(0)}=-c_e^{\rm biased}$) rather than
fine-tuning from an isotropic checkpoint is expected to let the encoder learn
the anisotropic metric directly, narrowing the gap to MWPM in that regime.

\section{Conclusion}
\label{sec:conclusion}

We have introduced a physics-informed GNN decoder for the surface code that
embeds the discrete Gauss law as a hard inductive bias, and we have identified
what carries the logical-error signal.  It is not a component of the edge
current: the current the solver returns is a gradient flow in its entirety, and
its harmonic component vanishes identically.  The logical content is a pairing
of the syndrome with the harmonic representative of the generator of
$H^1(G,\partial G;\Z_2)$---a representative that exists only once two boundary
sinks are declared, and that is fixed, among all members of its class, by the
learned metric.  We prove that the net current between two boundary sinks
placed on the pair of code boundaries linked by the logical operator evaluates
this pairing exactly and without free parameters
(Prop.~\ref{prop:pairing}), and that a single boundary node forfeits it by
merging the two boundaries, leaving no separating harmonic coordinate to pair
against.

Empirically, on the distance-$5$ rotated code under circuit-level noise, the
two-sink readout matches the best full-field readout---the node-potential
$\phi$-readout---in balanced accuracy (paired
$\Delta\mathrm{BA}=-0.012\pm0.024$ at $p=0.005$), and is statistically
indistinguishable from the single-sink current readout as well
($\Delta\mathrm{BA}=+0.005\pm0.012$).  A single topological
scalar thus attains the same accuracy as a full-field pooling readout, which is
precisely what a one-dimensional logical signal ($k=1$) predicts: projecting
onto the harmonic representative discards nothing.  The advantage of the harmonic
readout is not accuracy but interpretability, the absence of learned readout
parameters, and a principled route to
$k>1$ logical qubits.  The picture is preserved at larger code distance
($d=7$): the two-sink readout ($0.857\pm0.011$) significantly exceeds the
single-sink current pool (paired $\Delta\mathrm{BA}=+0.046\pm0.025$, $t=4.1$)
and ties the best full-field readout, so the closed-form scalar helps more, not
less, as the field it must pool grows---while never claiming an edge over the
$\phi$-readout, from which it is statistically indistinguishable.  The
comparison is made where the logical signal is well resolved; as $p$ grows the
signal weakens and all decoders degrade, so we claim no ordering in that regime.
The decoder does not surpass MWPM, which is near-optimal on the Nishimori
line.

The value of the construction is therefore not raw decoding performance but the
principled foundation it provides: by building the decoder on the discrete de
Rham complex of the syndrome graph, we obtain an exact, interpretable, and
differentiable physics layer whose structure is fixed by the code topology
($k$ harmonic representatives, one per logical qubit) and whose only learned
degrees of freedom are the noise-dependent metric $\bm{w}$.  An explicit harmonic
readout for codes with several logical qubits, and bias-aware training for the
regime where a learned decoder can genuinely surpass matching, are the natural
next steps.

More broadly, learned decoders with provable structural guarantees are an
emerging direction: beyond the surface code, related closed-form and
provably-scaling approaches have been developed for the decoding of quantum
states~\cite{Zhong2026}.  Whether a relative-cohomology readout of the kind
introduced here has an analogue in those settings---where the logical
information is likewise carried by a topological, rather than a learned,
structure---is an open question we leave to future work.

\begin{acknowledgments}
This work was supported by the Computational Science Center for Research Communities (CoSeC) NPRAISE project. P. E. T. acknowledges support from the CoSeC Collaborative Computational Project Quantum Computing (CCP-QC).
\end{acknowledgments}

\appendix

\section{Electromagnetic Reading of the Decomposition}
\label{app:em}

This appendix expands the electromagnetic analogy for the Hodge--Helmholtz
decomposition of Eq.~\eqref{eq:hhk}, referenced in Sec.~\ref{sec:onedim}.  It is
intuition, not a load-bearing part of the argument: the exactness of the readout
rests only on the graph being a $1$-dimensional complex (no co-exact sector),
established in the main text.

For a $1$-form in three dimensions this is the Helmholtz decomposition of a
vector field familiar from electromagnetism \cite{Frankel_2011}: a gradient
$\nabla f$ (exact, the irrotational component), a curl
$\nabla\times\mathbf{A}$ (co-exact, the solenoidal component, with
$\mathbf{A}$ the vector potential), and a harmonic remainder fixed by the
topology of the domain.  The analogy is one of mathematical structure and not
an identification: the electrostatic and magnetic fields are distinct physical
fields that happen to be respectively irrotational and solenoidal, whereas the
three terms above are components of a single decomposed field. \\
The middle, the solenoidal term is the one that requires a
$2$-form: writing $\nabla\times\mathbf{A}=\delta\bm{\beta}$ in the calculus of
forms, the potential $\bm{\beta}$ is a $2$-form (the Hodge dual of the vector
potential $\mathbf{A}$), and $\delta$ maps $2$-forms to $1$-forms.

Carrying the analogy one step further, one might reach for the Aharonov--Bohm
effect, which probes exactly this kind of object: a quantity invisible to any
local measurement, fixed instead by how the domain is connected.  The analogy
is suggestive but imprecise, and pinning down why sharpens what our readout
actually is.  The Aharonov--Bohm phase is the holonomy $\oint_\gamma\mathbf{A}$
of a connection around a \emph{closed} loop encircling a flux.  Our syndrome
graph is a $1$-dimensional complex with no co-exact sector, so the relevant
$1$-form $d_0\bm{\chi}$ carries no such holonomy: on the graph as an absolute
complex it is a gradient, and its integral between two vertices depends only on
the endpoints---there is no loop and no circulation.  What plays the role of the
loop is instead the \emph{relative} cycle---an error chain with both endpoints
on $\partial G$---which becomes closed only once the two code boundaries are
identified as a single element of $\partial G$.  The two equipotential sinks are
exactly that identification, and $s$ is the resulting relative pairing
$\langle d_0\bm{\chi},\bm{\sigma}\rangle$, evaluated at the boundary rather than
around a contour.  The electromagnetic picture that matches the construction
term for term is therefore not the Aharonov--Bohm loop but a more mundane one: a
resistor network with distributed current sources and two grounded terminals
held at a common potential, in which $s$ is the net current in the return path
between the two grounds.  The Aharonov--Bohm comparison is worth making only to
set it aside: it names the intuition---a globally defined, locally invisible
quantity---while the honest realisation on a graph is the grounded network.

The upshot for the main text is unchanged: on a graph the co-exact (solenoidal)
sector is absent, the current $\bm{J}$ splits into exact and harmonic parts
only, and the logical signal is the relative pairing of
Prop.~\ref{prop:pairing}, realised physically as the sink-current difference of
a grounded resistor network rather than as any contour holonomy.

\section{Jacobi Convergence Analysis}
\label{app:jacobi}

The Jacobi iteration on $\mathbf{L}\bm{\phi}=\bm{\rho}$ converges if and
only if the spectral radius of the iteration matrix
$\mathbf{I}-\mathbf{D}^{-1}\mathbf{L}$ is less than one, where
$\mathbf{D}=\mathrm{diag}(\mathbf{L})$.  For a weighted Laplacian with
$w_e>0$, the matrix $\mathbf{L}$ is symmetric positive semi-definite, and
$\mathbf{D}^{-1}\mathbf{L}$ has eigenvalues in $[0,2)$ for any graph
(since $\mathbf{L}$ is diagonally dominant).  Thus Jacobi always converges
for the Poisson equation with $w_e>0$, which is guaranteed by the $w_{\rm floor}$
regularisation.  With both boundary sinks held at zero, $\mathbf{L}$ restricted
to the free nodes is symmetric positive definite, so the Jacobi iteration
converges to the unique Dirichlet solution.

At the operating budget $K_{\rm J}=25$ the interior relative residual
$\|(\mathbf{B}^\top\bm{J}-\bm{\rho})_{v\notin\partial G}\|/\|\bm{\rho}\|$ is
$0.150\pm0.028$ across five seeds at $d=5$, $p=0.005$, falling geometrically to
$0.021\pm0.012$ by $K_{\rm J}=100$ (implied Jacobi contraction
$\rho\approx0.97$, time constant $\approx40$ iterations).  The decay confirms
that this is a genuine convergence residual: the solve at $K_{\rm J}=25$ is
stopped within the transient, not at the fixed point.

This is deliberate and does not affect the readout, for a reason worth stating
carefully.  The two-sink signal is the exact pairing of
Prop.~\ref{prop:pairing} \emph{at the fixed point}, but the decoder never
requires the fixed point.  Both training and inference read $s$ at the same
$K_{\rm J}=25$, so the quantity the classifier learns to calibrate is
$s^{(25)}(\bm{w})$: a deterministic, differentiable, and reproducible function
of the weights, whose value tracks the logical class whether or not the
underlying Poisson problem has fully converged.  Consistency is enforced by the
shared iteration budget, not by convergence to $s^{(\infty)}$; the residual is
truncation error on a quantity that is identically defined in training and in
test, and the same budget trains successfully at both $d=5$ and $d=7$.  Where
the exact pairing is wanted---for the explicit-harmonic validation of
Sec.~\ref{sec:future}---$\mathbf{L}$ restricted to the free nodes is symmetric
positive definite under the two-sink Dirichlet condition, so conjugate gradients
reaches $s^{(\infty)}$ in far fewer iterations than Jacobi; damped Jacobi
($\omega=2/3$) is the minimal change.

\section{Diffusion Stability Condition}
\label{app:diffusion}

Forward Euler integration of $\dot{\bm{u}}=-\mathbf{L}\bm{u}$ is stable if
and only if $\Delta t < 2/\lambda_{\max}(\mathbf{L})$.  For the syndrome
graph with $w_e\sim 1$ (after $w_{\rm floor}$ decay), a rough bound gives
$\lambda_{\max}\leq 2\max_v \sum_{e\ni v} w_e \leq 2\,d_{\max}\,w_{\max}$,
where $d_{\max}$ is the maximum degree.  Our constraint $\Delta t_{\max}=0.5$
is conservative with respect to this bound and ensures stability throughout
training.

\section{Training Loss Convergence}
\label{app:loss}

Figure~\ref{fig:loss} shows the per-epoch validation AUC for the four readout
architectures on the $d=5$ dataset, under the common training protocol of
Sec.~\ref{sec:training}.  All four readouts converge to a high AUC; the
diffusion readout, which suffers genuine gradient attenuation through the
Euler-discretised operator (Sec.~\ref{sec:readout_interp}), converges more
slowly and less smoothly but does reach a competitive value, confirming that the
vanishing gradient slows learning rather than preventing it.  This corrects the
picture from preliminary runs without warmup, in which the diffusion loss stalled
near the random level.

\begin{figure}[h]
\centering
\includegraphics[width=\textwidth]{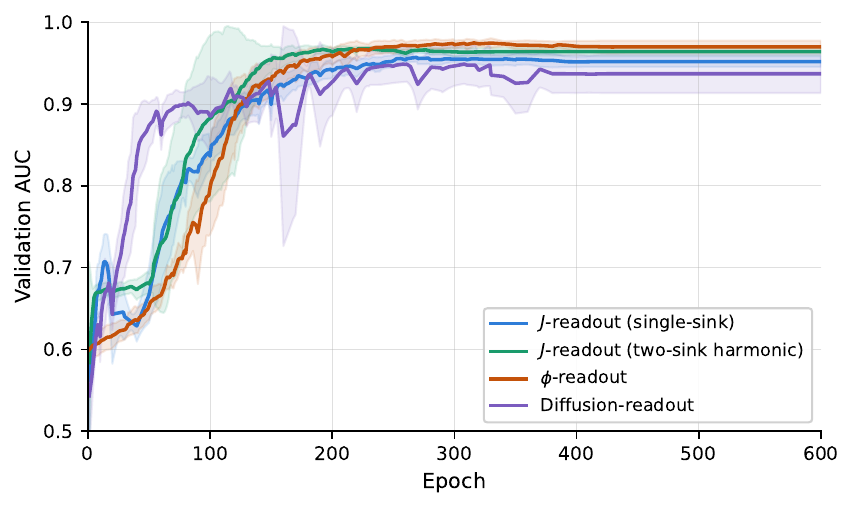}
\caption{%
  Validation AUC vs.\ epoch for the four readouts at $d=5$, $p=0.005$, under the
  common protocol of Sec.~\ref{sec:training} (mean over seeds, shaded band one
  standard deviation).  All four readouts train successfully; the diffusion
  readout converges more slowly but does not stall.  The dotted line marks the
  random baseline (AUC $=0.5$).
}
\label{fig:loss}
\end{figure}

\bibliography{refs}

\end{document}